\documentclass[11pt,letterpaper]{article}
\usepackage{authblk}
\usepackage[margin=1in]{geometry}
\usepackage[bottom]{footmisc}
\usepackage{amsmath,amssymb,amsfonts}
\usepackage{amsthm}
\usepackage{thm-restate}
\usepackage{thmtools}
\usepackage{microtype}
\usepackage{tikz}
\usepackage{float}
\newenvironment{claimproof}[1][Proof]{\begin{proof}[#1]}{\end{proof}}
\usepackage{xcolor}
\usetikzlibrary{positioning,calc}
\definecolor{cbBlue}{HTML}{0072B2}
\definecolor{cbOrange}{HTML}{E69F00}
\usepackage[hidelinks]{hyperref}
\usepackage[nameinlink,noabbrev]{cleveref}

\title{A Tight $\widetilde \Omega \left(\sqrt{m}\right)$ Information-Theoretic Lower Bound for \\
	Randomized Online Set Cover}
\author[1]{Ilan Doron-Arad}
\author[2]{Joseph (Seffi) Naor}
\affil[1]{MIT}
\affil[2]{Technion}
\date{}

\newcommand{\OPT}{\mathsf{OPT}}
\newcommand{\eps}{\varepsilon}

\theoremstyle{plain}
\newtheorem{theorem}{Theorem}[section]
\newtheorem{lemma}[theorem]{Lemma}

\newtheorem{claim}[theorem]{Claim}

\theoremstyle{definition}
\newtheorem{definition}[theorem]{Definition}
\newtheorem{question}[theorem]{Question}



\crefname{theorem}{theorem}{theorems}
\Crefname{theorem}{Theorem}{Theorems}
\crefname{lemma}{lemma}{lemmas}
\Crefname{lemma}{Lemma}{Lemmas}
\crefname{corollary}{corollary}{corollaries}
\Crefname{corollary}{Corollary}{Corollaries}
\crefname{claim}{claim}{claims}
\Crefname{claim}{Claim}{Claims}
\crefname{observation}{observation}{observations}
\Crefname{observation}{Observation}{Observations}
\crefname{definition}{definition}{definitions}
\Crefname{definition}{Definition}{Definitions}
\crefname{question}{question}{questions}
\Crefname{question}{Question}{Questions}

\begin{document}
	\hypersetup{pageanchor=false}
	
	\begin{titlepage}
		\maketitle
		\thispagestyle{empty}
		
		\begin{abstract}
			Online set cover is a fundamental problem in online algorithms, admitting a deterministic
			$O(\log m\log n)$-competitive algorithm, where $m$ is the number of sets and $n$ is the number of elements.
			This is essentially tight for deterministic algorithms as well as for polynomial-time randomized algorithms assuming $\mathrm{NP}\not\subseteq\mathrm{BPP}$.
			However, the best lower bound known for information-theoretic (computationally unlimited) randomized algorithms against an oblivious adversary is only $\Omega(\log m)$, whereas the upper bound in terms of $m$ is $O(\sqrt m\log m)=\widetilde O(\sqrt m)$.
			
			We prove an $\Omega(\sqrt m)$ lower bound for information-theoretic randomized unweighted online set cover, showing that even with unbounded computational power, randomization cannot achieve an $O(\log m)$-competitive ratio.
			Specifically, our lower bound rules out $O(\log m\cdot\log^{1/2-\eps} n)$-competitive algorithms for every constant $\eps>0$.
			Our techniques also prove an $\Omega(m^{1/3})$ lower bound in the random-order model, and show that every algorithm with $\textnormal{poly}(m)$ memory has competitive ratio $\Omega(m/\log m)$, even with unlimited computation between requests.
		\end{abstract}
	\end{titlepage}
	\hypersetup{pageanchor=true}
	
	\section{Introduction}
	
	Set cover is a textbook combinatorial optimization problem: given a universe $\mathcal U$ of $n$ elements and a family $\mathcal S$ of $m$ subsets of $\mathcal U$, the goal is to find a minimum-cardinality subfamily whose union is $\mathcal U$.
	Set cover is one of Karp's original $21$ NP-complete problems~\cite{Karp72}, and its best polynomial-time approximation factor is $\Theta(\log n)$ assuming $\textnormal{P} \neq \textnormal{NP}$, achievable by a simple greedy algorithm~\cite{Chvatal79,Johnson74,Lovasz75,Feige98,dinur2014analytical}.
	Denote by $\OPT$ the value of an optimal solution for a set cover instance.
	
	Online set cover, introduced by Alon, Awerbuch, Azar, Buchbinder, and Naor~\cite{AlonAwerbuchAzarBuchbinderNaor09}, has established itself as a fundamental problem in online algorithms.
	In the online case, at each time $t$, an element arrives; if it is not already covered, the algorithm must irrevocably buy a set covering the element.
	In competitive analysis, the guarantee of an online algorithm is measured against an optimal offline solution for the revealed sequence.
	For randomized algorithms, we consider expected competitiveness against an oblivious adversary who does not adapt to the algorithm's random choices and fixes the input sequence in advance. The online
	set cover problem has been studied in both the {\em known-instance model}, where the set system is known in advance, but not all elements will necessarily appear, and the {\em unknown-instance model}, where the set system is revealed online.
	
	Trying to match the offline $\Theta(\log n)$ bound also for online set cover may seem tempting; however, tight $\Theta(\sqrt n)$ bounds 
	already rule this out for both deterministic and randomized algorithms~\cite{AlonAwerbuchAzarBuchbinderNaor09,GuptaKehneLevin21}.
	Yet, there is a silver lining: by allowing the competitive ratio to also depend on the number of sets $m$, the work of~\cite{AlonAwerbuchAzarBuchbinderNaor09} obtained a deterministic $O(\log m\log n)$-competitive algorithm for the known-instance model and proved a lower bound of
	\(
	\Omega\!\left(\frac{\log m\log n}{\log\log n+\log\log m}\right)
	\)
	for deterministic algorithms, leaving only a gap of $\log\log n+\log\log m$ between them. In the unknown-instance model, randomized algorithms achieve an $O(\log m\log n)$ competitive ratio, but deterministic algorithms have separate $\Omega(m)$ and $\Omega(n)$ lower bounds even when $\OPT=1$~\cite{Korman04}; hence, our focus is the known-instance model. 
	
	For both models, the randomized algorithm of~\cite{AlonAwerbuchAzarBuchbinderNaor09} is based on the \emph{relax-and-round} framework: first, a deterministic $O(\log m)$-competitive algorithm is designed for fractionally solving online a linear programming relaxation of set cover; then, randomized rounding is applied online to this solution. For the known-instance model, it is also possible to apply an online de-randomization scheme, yielding the deterministic bound.
	As the integrality gap of the LP relaxation is $\Omega(\log n)$, and there is an $\Omega(\log m)$ lower bound on the fractional competitive ratio of online set cover, any improvement to the $O(\log m\log n)$ competitive factor of~\cite{AlonAwerbuchAzarBuchbinderNaor09} requires new ideas.
	
	In online algorithms, randomization often brings dramatically stronger guarantees compared to their deterministic counterparts, as in paging~\cite{FiatEtAl91}, or the $k$-server problem~\cite{BansalBuchbinderMadryNaor15,bubeck-k-server2018}.
	A partial resolution came in the work of Feige and Korman~\cite{Korman04}, who showed that under the standard complexity assumption $\mathrm{NP}\not\subseteq\mathrm{BPP}$ no polynomial-time randomized algorithm can achieve a competitive ratio of $o(\log m\log n)$ for online set cover.
	In other words, for polynomial-time algorithms, randomization is unlikely to help.
	
	While polynomial-time algorithms are a main focus in both offline and online algorithmic settings, the online setting fundamentally asks how well one can make decisions under an unpredictable future.
	From this viewpoint, understanding the best information-theoretic (i.e., computationally unlimited) competitive guarantee is of central importance.
	From this perspective, prior to our work, the strongest lower bound known for randomized online set cover, with no restriction on running time, was only $\Omega(\log m)$~\cite{AlonAwerbuchAzarBuchbinderNaor06,Korman04}.
	Gupta, Kehne, and Levin~\cite{GuptaKehneLevin21} also proved an $\Omega(\log m)$ lower bound that works even for random-order online set cover.\footnote{The bound in Theorem V.1 of~\cite{GuptaKehneLevin21} is stated as $\Omega(\log n)$. As the proof in the full version~\cite{gupta2021random} uses $m = 2^{\ell}$ and $n = 2^{\ell} - 1$, the resulting lower bound in terms of $m$ remains only $\Omega(\log m)$.} 
	Conversely, the current best upper bound for randomized information-theoretic algorithms is astonishingly larger: an $\widetilde O(\sqrt m)$-competitive algorithm by~\cite{Korman04}.\footnote{Korman~\cite{Korman04} gives an $O(\OPT\log(m/\OPT))$ upper bound; if $\OPT>\sqrt m$, any feasible algorithm is $\sqrt m$-competitive.}
	This leads us to the following basic question:
	
	\begin{question}
		\label{q:1}
		\emph{Is there an information-theoretic randomized algorithm for the online set cover problem in the known-instance setting that achieves an $O(\log m)$ competitive ratio? 
		}
	\end{question}
	We note that the lower-bound instances of~\cite{AlonAwerbuchAzarBuchbinderNaor09} satisfy $\OPT=1$. Interestingly, such instances are known to admit randomized $O(\log m)$-competitive algorithms~\cite{Korman04} (see Section~\ref{sec:prelim} for more details). Hence, any strengthening of the $\Omega(\log m)$ lower bound has to utilize instances for which $\OPT>1$. 
	
	\subsection{Our Results}
	
	Our main result is a qualitatively tight lower bound that gives a stark negative answer to \Cref{q:1}.
	We show that even with randomization and unbounded computational power, a polynomial-in-$m$ dependence of the competitive ratio is unavoidable. 
	
	\begin{restatable}{theorem}{main}
		\label{thm:randomized-lb-osc}
		Every randomized algorithm for unweighted online set cover has competitive ratio
		$\Omega(\sqrt m)$, even in the known-instance model.
	\end{restatable}
	
	Thus, the lower bound matches the $\widetilde O(\sqrt m)$ upper bound of~\cite{Korman04} up to a logarithmic factor.
	In particular, since $n=2^m-1$ in our construction, \Cref{thm:randomized-lb-osc} rules out any information-theoretic randomized $O(\log m\cdot\log^{1/2-\eps}(n))$-competitive algorithm that applies to all instances, for every constant $\eps>0$.
	Indeed, when $n=2^m-1$,
	$\log m\cdot\log^{1/2-\eps}(n)=o(\sqrt m)$, showing that either a dependence on $n$ or a polynomial dependence on $m$ are essential.

	It still remains open whether there is an information-theoretic randomized $O(\log(mn))$-competitive algorithm for adversarial-order online set cover.
	We give more details in Section~\ref{sec:discussion}. Note that obtaining an $O(\log(mn))$-competitive algorithm for arbitrary values of $m$ and $n$ does not contradict our $\Omega(\sqrt{m})$ lower bound, as in our construction we set $n = 2^m-1$, and thus $\log(mn) = \Theta(m)$. Hence, obtaining the latter bound is trivial in this regime. 
	
	\subsubsection*{Implications.}
	Numerous covering problems are solvable online via the primal-dual approach of~\cite{buchbinder2009design}.
	As mentioned before, for online set cover this leads to an $O(\log m)$-competitive guarantee for solving the LP relaxation fractionally, where $m$ refers to the number of variables.
	An additional factor of $\log n$, where $n$ is intuitively the number of constraints of the given problem, is incurred when rounding online the fractional solution and obtaining a feasible integral solution.
	Nevertheless, it has remained an open question whether a dependence on $n$ (or, equivalently, in the worst case, a polynomial dependence on $m$) is information-theoretically necessary for many problems.
	
	\Cref{thm:randomized-lb-osc} resolves this question, as it shows a polynomial, rather than logarithmic, dependence on the parameter analogous to $m$ for various online covering problems for which unweighted online set cover is either a special case, or there is a direct reduction from online set cover.
	In particular, this lower bound applies to online non-metric facility location~\cite{AlonAwerbuchAzarBuchbinderNaor06}, admission control~\cite{AlonAzarGutner09}, online node-weighted Steiner tree~\cite{NaorPanigrahiSingh11}, online prize-collecting Steiner problems~\cite{HajiaghayiLiaghatPanigrahi14}, online multistage matroid maintenance~\cite{GuptaTalwarWieder14}, online submodular cover~\cite{GuptaLevin20}, set cover with delay~\cite{AzarChiplunkarKuttenTouitou20}, and online weighted multi-level paging~\cite{BansalNaorTalmon21}.
	Additional examples of similar implications include~\cite{DBLP:conf/focs/AzarBCCCG0KNNP16,MertziosShalomWongZaks17,touitou2021nearly,DoronAradNaor24,BermanDasGupta08,BhawalkarGollapudiPanigrahi14,abshoff2016towards,borst2026buy,DoronAradGalNaor26}.
	
	Our techniques also give lower bounds for random order and space-efficient algorithms on which we now elaborate.

	\paragraph{Random-order.}
	In the random-order model, the adversary still chooses the set of element requests, but their arrival order is not adversarial: the requests are revealed according to a permutation chosen uniformly at random.
In this random-order model, Gupta, Kehne, and Levin~\cite{GuptaKehneLevin21} obtained a randomized $O(\log (mn))$-competitive algorithm. They also showed a lower bound of $\log m$ for $m \approx n$ and $\frac{\log m }{\log \log m}$ even when $m \gg n$. We focus on the complementary regime where $m \ll n$, and show that a polynomial dependence on $m$ (equivalently, a polylogarithmic dependence on $n$) is necessary in the random-order model.   
	
	
	\begin{restatable}{theorem}{randomorder}
		\label{thm:random-order-lb}
		Every randomized algorithm for unweighted online set cover in the known-instance random-order model has competitive ratio $\Omega\left(m^{1/3}\right)$. 
	\end{restatable}
	This lower bound holds 
	even when no two elements are covered by exactly the same collection of sets, i.e., without {\em element copies}. With element copies, the lower bound strengthens to $\Omega(\sqrt m)$ on instances with $n=\exp(O(m^{3/2}))$. 
	Since set cover is a special case of set multicover, pure covering integer programs, and online submodular cover, \Cref{thm:random-order-lb} gives an analogous lower bound for standard random-order versions of these problems~\cite{GuptaKehneLevin21,GuptaKehneLevin24,GuptaLevin20}.
	
	\paragraph{An $\widetilde \Omega (m)$-competitive lower bound for $\textnormal{poly}(m)$ space algorithms.}
    Ideally, in online algorithms, besides minimizing the competitive ratio, we seek algorithms whose memory is either small or independent of the length of the online request sequence. Such algorithms do exist for other fundamental online problems; for example, online (unweighted) paging~\cite{FiatEtAl91}, where the space and time complexity depend only on the cache size, apart from the logarithmic space needed to represent page identifiers, and thus does not grow with the length of the online sequence.
    
Thus, it is natural to measure the complexity of online set cover as a function of the number of available (fixed) resources $m$, and aim at avoiding a polynomial dependence of the memory on the sequence length (which can exponentially depend on $m$). Here, memory is the information preserved between requests; we assume that past requests cannot be reread if not stored in memory. It seems that this aspect of the online set cover problem has not been studied before. 

We show that in contrast to the information-theoretic setting, if the space complexity is only polynomial in the number of sets (i.e., $\textnormal{poly(m)}$), a competitive bound of $\widetilde O({\sqrt{m}})$ is no longer possible and the best competitive ratio that can be achieved is only $\widetilde \Omega (m)$. 
	
	\begin{restatable}{theorem}{spacelb}
		\label{thm:space-lb}
		Every randomized algorithm for unweighted online set cover that uses $\textnormal{poly}(m)$ memory has competitive ratio
		$\Omega(m/\log m) = \widetilde{\Omega}(m)$, even in the known-instance model and even with unlimited computation between consecutive requests.
	\end{restatable}

	\subsection{Our Techniques}
	\label{sec:techniques}
	As already mentioned, online set cover instances with $\OPT=1$ admit an $O(\log m)$-competitive randomized algorithm~\cite{Korman04} (see Section~\ref{sec:prelim} for more details).
	Thus, any strengthening of the randomized $\Omega(\log m)$ lower bound requires an instance for which $\OPT=k\gg1$.
	
	We view the unknown optimum as a ``hidden'' $k$-subset of the $m$ available sets.
	A clever algorithmic approach might try to keep track of all candidate $k$-subset covers that have {\em survived} so far (i.e., can still potentially be the optimal $k$-subset), and define a distribution over these candidate covers.
	The crux of this approach is that updating the distribution after each new request filters the surviving family.
	If this could be done by paying $O(1)$ extra sets for every halving of the surviving family, then an algorithm would pay overall $\log\binom{m}{k}\approx k\log m$.
	Conversely, if an algorithm can be forced to pay $\Omega(k)$ extra sets per halving, then it incurs a cost of $k\log\binom{m}{k}\approx k^2\log m$.
	As $\OPT=k$, the difference between the latter two bounds roughly corresponds to $\log m$ vs. $k\log m$ competitive guarantees, which is substantial for $k=\operatorname{poly}(m)$.
	
	We show that for any randomized algorithm, there exists an online set cover instance that forces the algorithm to incur a high cost for covering it.
	A \emph{candidate cover} is a $k$-subset of the $m$ available set indices, and a \emph{candidate family} is a family of candidate covers.
	The elements of the online set cover instance are given by an adversary in batches.
	For each nonempty candidate family, we construct a designated {\em forcing batch} such that every bounded-size purchased family 
    covers the batch if and only if it fully contains at least one candidate cover from the candidate family.
	Such a contained candidate cover serves as a \emph{witness} for the current candidate family.
	
	Starting from a large candidate family $\mathcal C$, the adversary proceeds in rounds and in every round retains a uniformly random half of the current candidate family, independently of the algorithm's choices.
	It then presents the forcing batch for the retained half.
	Thus, the current witness is discarded by the next random half with probability $1/2$.
	Consequently, after many rounds, the algorithm is forced to purchase many new witnesses.
	Since the algorithm's choices are irrevocable, all previously chosen witnesses remain in its final purchased family.
	
	The witnesses obtained over the rounds need not be disjoint.
	A main combinatorial challenge is thus to guarantee that their union still contains a large number of set indices.
	To this end, we use a large \emph{union-expanding} family, in which the union of any $q\approx\sqrt m$ distinct candidate covers is at least $r=\Theta(m)$.
	Therefore, if at least $q$ distinct witnesses are obtained, the algorithm buys more than $r$ sets.
	In contrast, the unique candidate cover that survives the entire filtration covers every batch and has size $k\approx\sqrt m$.
	Putting these ingredients together yields the $\Omega(\sqrt m)$ lower bound.
	
	We give a probabilistic construction of the required union-expanding family.
	The same set-system property is related to a hypergraph formulation studied since Brown, Erd\H{o}s, and S\'os~\cite{BrownErdosSos73}; see also~\cite{alon2006extremal,shangguan2023degenerate,delcourt2024limit}.

\paragraph{Random-order.} For the random-order result, the main idea is to force the random-order to be close to the worst-case construction above. Using sufficiently many element copies of each original element (of the worst-case construction), one can force with high probability essentially any desired order on the original elements. To obtain a lower bound without element copies, we instead use a subset of the sets $[d]\subseteq[m]$ to simulate the worst-case hardness construction, while the remaining sets are used to distinguish between elements that would otherwise have the same incidence pattern.

Specifically, we consider each filtration step of the adversarial construction, in which we halve the current candidate family, as a {\em level}. For each level and every subset $A\subseteq[d]$, we create many corresponding requests. This is necessary because the candidate-family filtration is random, so before the filtration is sampled we do not know at which level a given incidence pattern $A$ will be needed. The number of requests per level decreases geometrically, so earlier levels contain many more requests than later ones. Thus, with high probability, sufficiently many requests from each level are revealed before any request from a later level, reproducing the sequence of forcing batches from the adversarial-order construction. 
This allows us to apply essentially the same filtration argument, yielding expected cost $\Omega(d)$, while $\OPT=O(\sqrt d)$. Distinguishing between all requests without using element copies requires $\Theta(d^{3/2})$ additional sets, so we choose $d=\Theta(m^{2/3})$, giving the claimed $\Omega(m^{1/3})$ lower bound.

    \paragraph{Bounded space.} The proof of \Cref{thm:space-lb} uses a reduction from $q$-party one-way set-disjointness. In this problem, each player $i$, $1\le i\le q$, receives a set $S_i\subseteq[N]$, under the promise that either the sets $S_1,\ldots,S_q$ are pairwise disjoint, or there is a unique element contained in all $q$ sets. Every randomized one-way protocol that distinguishes between these two cases with constant probability, requires $\Omega(N/q)$ communication between consecutive players by the work of Chakrabarti, Khot, and Sun~\cite{ChakrabartiKhotSun03}.

In our reduction to online set cover, each element $s$ of the disjointness universe is represented by a candidate cover $W_s$ from a union-expansion family (with a similar proof, yet different parameters from the information-theoretic lower bound). Thus, each player replaces its input set $S_i$ by the corresponding candidate family $\{W_s:s\in S_i\}$ and presents its forcing batch to the online algorithm. Between consecutive players, we can only pass the memory state of the online algorithm, so a $\operatorname{poly}(m)$-space algorithm induces a low-communication protocol. If the sets $S_1,\ldots,S_q$ are pairwise disjoint, the forcing batches force the algorithm to purchase distinct candidate covers (though not necessarily disjoint), and union-expansion forces a large union, i.e., large cost; conversely, if all sets share an element $s^\star$, then $W_{s^\star}$ covers every forcing batch, so $\OPT$ is small. Thus, a better competitive ratio would yield a low-communication protocol, contradicting the lower bound of~\cite{ChakrabartiKhotSun03}.

	\paragraph*{Organization.}
	\Cref{sec:prelim} provides background on the online set cover problem.
	In \Cref{sec:components}, we define the set cover instance used in the lower bound proofs, as well as the main combinatorial ingredients of the proofs: forcing batches and union-expanding candidate families.
	In \Cref{sec:main_hardness}, we prove our main worst-case lower bound, the random-order lower bounds, and the polynomial-space lower bound.
	\Cref{sec:discussion} concludes with a discussion of open questions and future research directions.
	
	\section{Preliminaries}
	\label{sec:prelim}
	We start with a more concrete definition of the problem and model.

	\begin{definition}[Online set cover (unweighted)]
		Let $\mathcal U$ be a finite universe and let $\mathcal S\subseteq2^{\mathcal U}$ be a family of subsets.
		We write $n:=|\mathcal U|$ and $m:=|\mathcal S|$.
		An instance is given by $(\mathcal U,\mathcal S)$ together with an online sequence of elements $e_1,e_2,\ldots,e_T\in\mathcal U$ revealed one-by-one.
		We refer to a contiguous subsequence of the elements as a \emph{batch}.

		An online algorithm maintains a sequence of purchased subfamilies
		$X_0\subseteq X_1\subseteq\cdots\subseteq X_T\subseteq\mathcal S$
		such that, for every $t\in[T]$, all elements revealed up to time $t$ are covered, i.e.,
		\(
		\{e_1,\ldots,e_t\}\subseteq\bigcup_{S\in X_t}S.
		\)
		The cost of the algorithm is $| X_T|$.
		The offline optimum is
		\[
		\OPT:=\min\left\{| X|: X\subseteq\mathcal S\text{ and }\{e_1,\ldots,e_T\}\subseteq\bigcup_{S\in X}S\right\}.
		\]
	\end{definition}
	
	A deterministic algorithm is $\alpha$-{\em competitive} if $| X_T|\le\alpha \cdot \OPT$ for every input sequence.
	A randomized algorithm is $\alpha$-competitive against an {\em oblivious} adversary if $\mathbb E[| X_T|]\le\alpha \cdot \OPT$ for every fixed (in advance) input sequence, where expectation is taken over the algorithm's random choices.
	
	\paragraph*{Models and adversaries.}
	Two information models for online set cover are considered. 
	\begin{itemize}
		\item In the {\em known-instance} model, the set system $(\mathcal U,\mathcal S)$ is known to the algorithm in advance. The order in which the elements appear is unknown, and some elements may not be requested. 
		\item 	In the {\em unknown-instance} model, the algorithm initially knows only the $m$ set labels, and the sets covering each element are revealed upon its arrival. 
	\end{itemize}
	Since our lower bound holds in the known-instance model, it also applies to the unknown-instance model. We also note that against an {\em adaptive} adversary (rather than oblivious), an $\Omega(m)$ lower bound is immediate even when $\OPT=1$.

	\paragraph*{Random-order online set cover.}
	The set system $(\mathcal U,\mathcal S)$ is given to the algorithm in advance. The adversary fixes a set $E\subseteq\mathcal U$ of $N$ distinct requested elements, and the value of $N$ may also be given to the algorithm. The elements of $E$ are then revealed in a uniformly chosen random order. Distinct elements may belong to the exact same collection of sets, and then they are called {\em copies}. A randomized algorithm is $\alpha$-competitive in this model if, for every fixed $(\mathcal U,\mathcal S,E)$,
	\(
	\mathbb E_{\pi,\mathcal A}[| X_N|]\le\alpha \cdot \OPT,
	\)
	where $\pi$ is a uniformly chosen random permutation of $E$ and the remaining randomness is internal to the algorithm.

	\paragraph*{The case $\OPT=1$.}
	For completeness, we give the $O(\log m)$-competitive algorithm for $\OPT = 1$~\cite{Korman04}. 
	Consider an instance whose revealed sequence can be covered by one available set.
	The algorithm maintains the family of sets that cover every element requested so far and, whenever its currently selected set ceases to belong to this family, samples a new set uniformly from the surviving family.
	Against an oblivious adversary, the selected set is uniform in the current surviving family.
	If the successive distinct sizes of the surviving family are
	$m=a_0>a_1>\cdots>a_\ell\ge1$, then, including the initial purchase, the expected number of purchases is at most
	\[
	1+\sum_{j=1}^{\ell}\frac{a_{j-1}-a_j}{a_{j-1}}
	\le 1+\sum_{h=2}^{m}\frac1h
	=O(\log m).
	\]
	The inequality holds because, for every $j$, each term $1/h$ with $a_j<h\le a_{j-1}$ is at least $1/a_{j-1}$, and these intervals are disjoint.

	\paragraph*{Notation.} Throughout, $[m]:=\{1,2,\ldots,m\}$ and, for $0 \le k \le m$,
	$\binom{[m]}{k}:=\{A\subseteq[m]:|A|=k\}$.
	All logarithms are natural unless stated otherwise.

	\section{The Combinatorial Constructions}
	\label{sec:components}
	In this section, we define the set system used in our lower bound and give its two main combinatorial ingredients: forcing batches for arbitrary nonempty candidate families and union-expanding candidate families.
	
	\paragraph*{The set system.}
	Fix $m\ge1$. For every nonempty subset $A\subseteq[m]$, introduce a distinct universe element $u_A$, whose {\em incidence pattern} is $A$, i.e.,
	$A$ is the subcollection of sets covering $u_A$.
	The universe is
	\(
	\mathcal U_m:=\{u_A \mid A\subseteq[m], A \neq \emptyset\},
	\)
	and, for each $i\in[m]$, the $i$-th available set is
	\(
	S_i:=\{u_A \mid i\in A\}.
	\)
	Thus, $|\mathcal U_m|=2^{m}-1$ and there are exactly $m$ available sets.
	We identify a purchased subfamily with its index set $X\subseteq[m]$.
	By construction,
	\begin{equation}
		\label{eq:incidence}
		X\text{ covers }u_A\quad\Longleftrightarrow\quad X\cap A\ne\emptyset.
	\end{equation}
	A \emph{candidate cover} is a subset $W\subseteq[m]$, and a \emph{candidate family} is a collection of candidate covers; in the lower-bound construction, all candidate covers have size $k$. 

	\paragraph{Forcing batches.} For every $\mathcal D\subseteq2^{[m]}$ and every integer $r \ge 1$, we say that a sequence of elements $F_r(\mathcal D)$ from $\mathcal{U}_m$ is an {\em $r$-forcing batch} of $\mathcal{D}$ if for every
	$|X|\le r$,
	\[
	X\text{ covers every element of }F_r(\mathcal D)
	\quad\Longleftrightarrow\quad
	\text{there exists }W\in\mathcal D\text{ with }W\subseteq X.
	\]
	The next lemma provides the forcing batch used in the construction. 
	
	\begin{restatable}{lemma}{forcing}
		\label{lem:forcing-batch}
		For every nonempty candidate family $\mathcal D\subseteq2^{[m]}$ consisting of nonempty candidate covers and every integer $r\ge1$, there is an $r$-forcing batch of  $\mathcal D$ such that $|F_r(\mathcal D)|\le\sum_{j=0}^{r}\binom mj$. 
	\end{restatable}
	
	\begin{proof}
		Define the forcing batch by
		\[
		F_r(\mathcal D):=\Big(u_A\ \Big|\ A\subseteq[m],\ \big|[m]\setminus A\big|\le r,\text{ and }A\cap W\ne\emptyset ~\forall W\in\mathcal D \Big),
		\]
		with the elements listed according to a fixed arbitrary rule on $2^{[m]}$. Let $X\subseteq[m]$ satisfy $|X|\le r$.
		If $W\subseteq X$ for some $W\in\mathcal D$, then since $u_A\in F_r(\mathcal D)$ satisfies $A \cap W \neq \emptyset$ it follows $A\cap X\ne\emptyset$, so $X$ covers $F_r(\mathcal D)$ by~\eqref{eq:incidence}.
		Conversely, if no $W\in\mathcal D$ is contained in $X$, then $A:=[m]\setminus X$ intersects every $W\in\mathcal D$ and satisfies $|[m]\setminus A|=|X|\le r$. Thus, $u_A\in F_r(\mathcal D)$, whereas $X\cap A=\emptyset$.
		The size bound follows since every element in $F_r(\mathcal D)$ is indexed by the complement of a set of size at most $r$.
	\end{proof}
	\paragraph{Union-expanding candidate families}
	
	We will use a large family with the property that the union of any $q$ distinct candidate covers is large.
	This is the combinatorial property needed in the lower-bound proof.
	
	\begin{definition}
		\label{def:union-expansion}
		Let $m,k,q,r \in \mathbb{N}$ with $1\le k\le r\le m$ and $q\ge2$.
		A candidate family $\mathcal C\subseteq\binom{[m]}{k}$ is called $(m,k,q,r)$-\emph{union-expanding} if every $q$ distinct candidate covers $W_1,\ldots,W_q\in\mathcal C$ satisfy
		$
		\left|\bigcup_{i=1}^{q}W_i\right|>r.
		$
	\end{definition}
	
	We first use the first moment method to obtain a general criterion for union-expansion. 
	
	\begin{lemma}
		\label{lem:union-expansion-lower}
		Let $m,k,q,r$ be integers with $1\le k\le r\le m$ and $q\ge2$, and let $M$ be an integer with $0\le M\le\binom{m}{k}$. If
		\[
		\binom{M}{q}\binom{m}{r}
		\left(\frac{\binom{r}{k}}{\binom{m}{k}}\right)^q<1,
		\]
		then there is an $(m,k,q,r)$-union-expanding candidate family of size $M$.
	\end{lemma}
	
	\begin{proof}
		If $M<q$, the assertion is vacuous. If $\binom{r}{k}<q$, no $r$-set contains $q$ distinct $k$-sets, so any $M$-element family works. Otherwise, choose uniformly an $M$-element family $\mathcal C\subseteq\binom{[m]}k$, and let $Z$ count the pairs $(R,Q)$ such that $R\in\binom{[m]}r$ and $Q\subseteq\mathcal C$ consists of $q$ members contained in $R$. By linearity of expectation,
		\[
		\begin{aligned}
			\mathbb E[Z]
			&=\binom mr
			\binom{\binom rk}{q}
			\frac{\binom Mq}{\binom{\binom mk}{q}}\\
			&=\binom Mq\binom mr
			\prod_{j=0}^{q-1}
			\frac{\binom rk-j}{\binom mk-j}\\
			&\le\binom Mq\binom mr
			\left(\frac{\binom rk}{\binom mk}\right)^q<1.
		\end{aligned}
		\]
		The inequality holds since $\frac{\binom rk-j}{\binom mk-j} \leq \frac{\binom rk}{\binom mk}$ for every $0 \le j \le q-1$.
		Thus, some choice of $\mathcal C$ satisfies $Z=0$. If $q$ of its members had a union of size at most $r$, their union could be extended to an $r$-set, contradicting $Z=0$.
	\end{proof}
	
	We now make the criterion specific to the parameters used in our proofs. 
	
	\begin{restatable}{lemma}{unionA}
		\label{thm:union-expansion-main}

		For all sufficiently large $m \in \mathbb{N}$ 
		there is an $(m,k,q,r)$-union-expanding candidate family of size $M$ for both of the following parameter regimes.
		\begin{enumerate}
			\item $k:=\lfloor\sqrt m\rfloor$, $q:=\lceil\sqrt m\rceil$, $r:=\lfloor m/32\rfloor$, and $M:=2^{4q}$. 
			\item For every fixed $d \in \mathbb{N}$, $
			k:=\left\lceil2(d+3)\log m\right\rceil$
			$q:=\left\lfloor\frac{m}{k}\right\rfloor$, $r:=\left\lfloor\frac m{32}\right\rfloor$, and $M:=m^{d+3}+q$. 
		\end{enumerate}
		
	\end{restatable}
	
	\begin{proof}

		Let $H(x):=-x\log x-(1-x)\log(1-x)$ be the binary entropy function. 
		For both parameter regimes it holds that
		\begin{equation}
			\label{eq:union-expansion-bound}
			\begin{aligned}
				&\log\left[\binom Mq\binom mr
				\left(\frac{\binom rk}{\binom mk}\right)^q~\right]\\
				&\qquad=\log \binom Mq + \log \binom mr + q \log \frac{\binom rk}{\binom mk}\\
				&\qquad\le q\log\frac{eM}{q}+mH(r/m)+qk\log(r/m).\\
			\end{aligned}
		\end{equation}
		The inequality uses the following known facts. First, $\binom M q\le(e M/q)^q$ for all $1 \le q \le M$; second, $\binom mr\le e^{mH(r/m)}$ for all $0 \le r \le m$; third, since $\frac{r-j}{m-j} \leq \frac{r}{m}$ for every $0 \le j \le m-1$, we have
		\[
		\frac{\binom rk}{\binom mk}
		=\prod_{j=0}^{k-1}\frac{r-j}{m-j}
		\le\left(\frac rm\right)^k.
		\]
		
		To conclude, we consider the two parameter regimes separately.
		\begin{enumerate}
			\item $k:=\lfloor\sqrt m\rfloor$, $q:=\lceil\sqrt m\rceil$, $r:=\lfloor m/32\rfloor$, and $M:=2^{4q}$. Then, $$\log\left[\binom Mq\binom mr
			\left(\frac{\binom rk}{\binom mk}\right)^q~\right] \le \left(4\log2+H(1/32)-\log32+o(1)\right)m <0.$$
			The first inequality follows from \eqref{eq:union-expansion-bound} and since $q^2=m+O(\sqrt m)$, $qk=m+O(\sqrt m)$, and $q\log q=o(m)$. 
			
			\item For every fixed $d \in \mathbb{N}$, $
			k:=\left\lceil2(d+3)\log m\right\rceil$,
			$q:=\left\lfloor\frac{m}{k}\right\rfloor$, $r:=\left\lfloor\frac m{32}\right\rfloor$, and $M:=m^{d+3}+q$. Then, 
			
			$$\log\left[\binom Mq\binom mr
			\left(\frac{\binom rk}{\binom mk}\right)^q~\right]\le\left(\frac12+H(1/32)-\log32+o(1)\right)m <0.$$
			The first inequality follows from \eqref{eq:union-expansion-bound} and since $M=m^{d+3}(1+o(1))$, $q\le m/k$, $qk=m-O(k)$, and $r/m=1/32+o(1)$.
		\end{enumerate}

		It follows that $\binom{M}{q}\binom{m}{r}
		\left(\frac{\binom{r}{k}}{\binom{m}{k}}\right)^q<1$ in both parameter regimes. 
		Additionally, for both parameter regimes, for sufficiently large $m$, it holds that
		$
		\log\binom mk\ge k\log(m/k)>\log M,
		$
		so $M\le\binom mk$. 
		Thus, the proof follows from \Cref{lem:union-expansion-lower}.
	\end{proof}
	
	\section{The Hardness of Randomized Online Set Cover}
	\label{sec:main_hardness}
	
	In this section, we first prove our main lower bound, \Cref{thm:randomized-lb-osc}. Then, in \Cref{sec:random-order-proof} we prove the random-order lower bound, and in \Cref{sec:space-proof} we prove the polynomial-space lower bound.
	Our lower bounds utilize Yao's minimax principle; thus, we start by constructing a distribution over request sequences. The distribution is obtained from the set system introduced in the previous section and relies on the following random filtration system.
	
	Fix integers $1\le k\le r\le m$ and $q\ge2$, and set $T:=4q$. Let $\mathcal C^{(0)}\subseteq\binom{[m]}k$ be an $(m,k,q,r)$-union-expanding candidate family of size $2^T$.
	For every $t<T$, choose $\mathcal C^{(t+1)}$ uniformly at random from among all possible halves of $\mathcal C^{(t)}$. 
	Let $X_0\subseteq\cdots\subseteq X_T\subseteq[m]$ be random sets such that, for every $t<T$, $X_t$ is determined after $\mathcal C^{(t)}$ is chosen and before $\mathcal C^{(t+1)}$ is chosen. The sets $X_t$ can be thought of as the sets purchased up to time $t$ by a fixed deterministic algorithm, so as to satisfy a given request sequence. 
    
    The goal of the next lemma is to provide a lower bound on $\mathbb E[|X_T|]$. 
    The lemma is formulated in a way that allows it 
    to be later used in \Cref{sec:random-order-proof} for the random-order lower bound. The
    proof uses the simple observation that, for any fixed $W\in\mathcal C^{(t)}$, $W\notin\mathcal C^{(t+1)}$ with probability $1/2$, since $\mathcal C^{(t+1)}$ is a uniformly chosen random half of $\mathcal C^{(t)}$.

	\begin{lemma}
		\label{lem:filtration}
		If for every $t < T$ such that $|X_t|\le r$ there is $W_t\in\mathcal C^{(t)}$ satisfying $W_t\subseteq X_t$, then $\mathbb E[|X_T|]\ge(1-e^{-q/4})r$.
	\end{lemma}

\begin{proof}

For each \(t<T\), choose \(W_t\in\mathcal C^{(t)}\), based only on the history up to time \(t\), such that \(W_t\subseteq X_t\) whenever \(|X_t|\le r\), which is possible by the lemma statement, choose $W_t$ arbitrarily otherwise.
    Define
	$
	I_t:=\mathbf 1[W_t\notin\mathcal C^{(t+1)}].
	$
Since $W_t$ belongs to $\mathcal C^{(t+1)}$ with probability $\frac{1}{2}$, it holds that 
$
\Pr[I_t=1]=\frac12.
$
	Moreover, for every $t<T$ and every fixed $a_0,\ldots,a_{t-1}\in\{0,1\}$,
	 since $I_0,\ldots,I_{t-1}$ are determined before
	 $\mathcal C^{(t+1)}$ is chosen,
	\[
	\Pr[I_t=1\mid I_0=a_0,\ldots,I_{t-1}=a_{t-1}]=\frac12.
	\]
	Therefore, $I_0,\ldots,I_{T-1}$ are independent
	$\operatorname{Bernoulli}(1/2)$ random variables.
	Since $T=4q$, we have
	$
	\mu:=\mathbb E\!\left[\sum_{t<T} I_t\right]=\frac{T}{2}=2q.
	$
	Applying the multiplicative Chernoff bound~\cite{MitzenmacherUpfal17} with
	$\delta=1/2$ gives
	\begin{equation}
		\label{eq:cher}
			\Pr\!\left[\sum_{t<T} I_t<q\right]
		=
		\Pr\!\left[\sum_{t<T} I_t<(1-\delta)\mu\right]
		\le
		\exp\!\left(-\frac{\delta^2\mu}{2}\right)
		=
		e^{-q/4}.
	\end{equation}

	We prove that if $|X_T|\le r$ then $\sum_{t<T} I_t<q$. Assume that $|X_T|\le r$; then, 
	 $|X_t|\le r$ for every $t<T$ since $X_t \subseteq X_T$, implying
$
	W_t\subseteq X_t\subseteq X_T.
	$
	Moreover, if $I_t=1$ for some $t<T$, then for any $s>t$ by definition of $W_s$ we have
	$W_s\in\mathcal C^{(s)}\subseteq\mathcal C^{(t+1)}
$; thus, necessarily $W_t\ne W_s$. Since $X_T$ cannot contain $q$ distinct members of
	$\mathcal C^{(0)}$ by the $(m,k,q,r)$-union-expansion property, it follows that $\sum_{t<T} I_t<q
	$. Thus,
	$
	|X_T|\le r
	$ implies
 	$\sum_{t<T} I_t<q.
	$
	Consequently, combined with \eqref{eq:cher},
	\[
	\Pr[|X_T|\le r]
	\le
	\Pr\!\left[\sum_{t<T}I_t<q\right]
	\le e^{-q/4}.
	\]
	Therefore,
	\[
	\mathbb E[|X_T|]
	\ge
	r\cdot \Pr[|X_T|>r]
	\ge
	(1-e^{-q/4})r,
	\]
	as required.
\end{proof}

	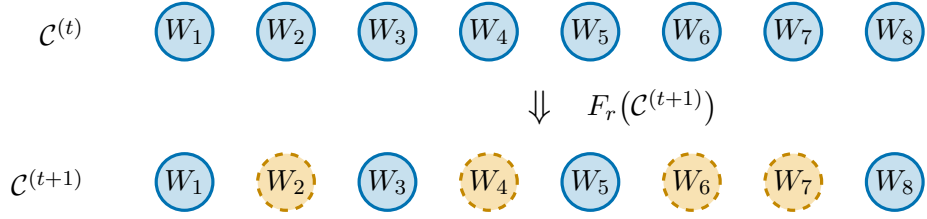
\begin{figure}[H]
		\centering
		\begin{tikzpicture}[
			candidate/.style={
				circle,
				draw=cbBlue,
				very thick,
				fill=cbBlue!22,
				minimum size=7.5mm,
				inner sep=0pt
			},
			retained/.style={
				circle,
				draw=cbBlue,
				very thick,
				fill=cbBlue!22,
				minimum size=7.5mm,
				inner sep=0pt
			},
			discarded/.style={
				circle,
				draw=cbOrange!85!black,
				very thick,
				dashed,
				fill=cbOrange!30,
				minimum size=7.5mm,
				inner sep=0pt
			}
			]
			
			\node[candidate] (t1) {$W_1$};
			\node[candidate, right=5.5mm of t1] (t2) {$W_2$};
			\node[candidate, right=5.5mm of t2] (t3) {$W_3$};
			\node[candidate, right=5.5mm of t3] (t4) {$W_4$};
			\node[candidate, right=5.5mm of t4] (t5) {$W_5$};
			\node[candidate, right=5.5mm of t5] (t6) {$W_6$};
			\node[candidate, right=5.5mm of t6] (t7) {$W_7$};
			\node[candidate, right=5.5mm of t7] (t8) {$W_8$};
			
			\node[retained, below=12mm of t1] (b1) {$W_1$};
			\node[discarded, right=5.5mm of b1] (b2) {$W_2$};
			\node[retained, right=5.5mm of b2] (b3) {$W_3$};
			\node[discarded, right=5.5mm of b3] (b4) {$W_4$};
			\node[retained, right=5.5mm of b4] (b5) {$W_5$};
			\node[discarded, right=5.5mm of b5] (b6) {$W_6$};
			\node[discarded, right=5.5mm of b6] (b7) {$W_7$};
			\node[retained, right=5.5mm of b7] (b8) {$W_8$};
			
			\node[left=8mm of t1] {$\mathcal C^{(t)}$};
			\node[left=8mm of b1] {$\mathcal C^{(t+1)}$};
			
			\node at ($(t4)!0.5!(t5)+(0,-10mm)$) (arrow) {\Large$\Downarrow$};
			\node[right=2mm of arrow] {$F_r\!\left(\mathcal C^{(t+1)}\right)$};
		\end{tikzpicture}
		\caption{One step of the random filtration process. The family $\mathcal C^{(t+1)}$ is a uniformly random half of $\mathcal C^{(t)}$. Solid blue indicates candidate covers in $\mathcal C^{(t+1)}$ and dashed orange indicates discarded candidates from $\mathcal C^{(t)} \setminus \mathcal C^{(t+1)}$. The forcing batch $F_r(\mathcal C^{(t)})$ (not shown here) forces the algorithm to contain at least one of $W_1,\ldots,W_8$. Then, after  
			the next forcing batch $F_r(\mathcal C^{(t+1)})$ arrives, it ensures that the algorithm contains at least one of the covers $W_1,W_3,W_5,W_8$. Thus, with probability $1/2$ the algorithm must choose a new candidate cover. 
		}
		\label{fig:random-halving}
	\end{figure}

	We can now prove our main result. 
	
	\main*
	
	\begin{proof}
		Fix a sufficiently large $m$, and set
		$k:=\lfloor\sqrt m\rfloor$, $q:=\lceil\sqrt m\rceil$, $r:=\lfloor m/32\rfloor$, and $T:=4q$. 
		By \Cref{thm:union-expansion-main}, fix an $(m,k,q,r)$-union-expanding family $\mathcal C^{(0)}\subseteq\binom{[m]}{k}$ of size $2^T$, and choose every $\mathcal C^{(t+1)}$ uniformly at random among the halves of $\mathcal C^{(t)}$. Thus, $\mathcal C^{(T)}=\{W^\star\}$ for a single candidate cover. Reveal the forcing batches
		$
		F_r(\mathcal C^{(0)}),F_r(\mathcal C^{(1)}),\ldots,F_r(\mathcal C^{(T)}).
		$
		Once the filtration system $\left(\mathcal{C}^{(t)}\right)_{t \le T}$ is sampled, the request sequence is fixed, so the filtration process defines a distribution over oblivious inputs. A visualization of a single filtration step is given in Figure~\ref{fig:random-halving}. 
		
		By Yao's minimax principle, fix a deterministic algorithm. Let $X_t\subseteq[m]$ be all set indices purchased by the algorithm up to and including the batch $F_r(\mathcal C^{(t)})$. 
		 Then, $X_0\subseteq\cdots\subseteq X_T$, and conditioning on the history up to time $t$, it holds that $\mathcal C^{(t+1)}$ is a uniformly random half of $\mathcal C^{(t)}$. For every $t<T$ such that $|X_t|\le r$, \Cref{lem:forcing-batch} gives a candidate $W_t\in\mathcal C^{(t)}$ with $W_t\subseteq X_t$. Hence, \Cref{lem:filtration} gives $\mathbb E[|X_T|]\ge(1-e^{-q/4})r$.
		Since $W^\star\in\mathcal C^{(t)}$ for every $t$, the sets indexed by $W^\star$ cover every batch, and therefore $\OPT\le k$. Yao's minimax principle then gives competitive ratio
		$
		\frac{(1-e^{-q/4})r}{k}=\Omega(\sqrt m)
		$. 
	\end{proof}

	\subsection{Random-Order Lower Bound}
	\label{sec:random-order-proof}

We next prove our random-order result, using the level construction described in Section~\ref{sec:techniques}. Recall that each level corresponds to one filtration step of the adversarial-order construction. For every level $t$ and nonempty $A\subseteq[d]$, we create $L_t$ corresponding requests $(t,A,\ell)$, where $L_t$ decreases geometrically with $t$. Request $\ell$ at level $t$ with incidence $A\subseteq[d]$ is kept distinct from other requests at $t$ and $A$ by adding a unique additional incidence pattern from the additional sets $R_{t,A,\ell}\subseteq[m]\setminus[d]$. An illustration is given in \Cref{fig:random-order-levels}. 
Then, a uniformly random ordering of all requests is generated by independently assigning each request a priority $p_{t,A,\ell}\sim\textnormal{Unif}[0,1]$ and ordering the requests according to increasing priority values.
For each level $t$, we define a priority threshold $z_t \in [0,1]$, so that with high probability, by processing all requests with priority at most $z_t$, sufficiently many requests from level $t$ will have arrived to enforce the forcing batch of a corresponding family $\mathcal{C}^{(t)}$, as in the adversarial construction, while no request from any later level $t'>t$ will have arrived (which would ruin the filtration process). 

We first show that, with high probability, the thresholds reproduce the desired filtration structure. Conditioning on this event, the adversarial-order filtration lemma yields the cost lower bound, while the surviving candidate cover keeps $\OPT$ small.

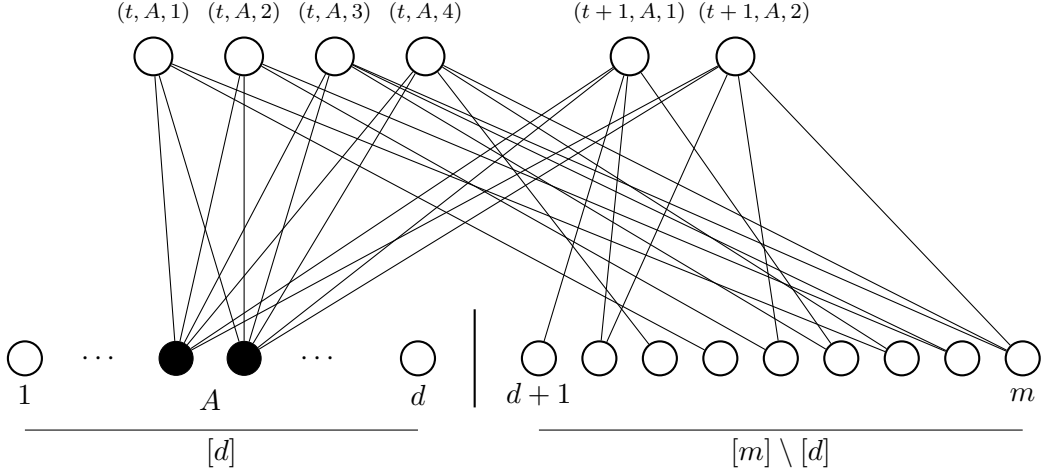
\begin{figure}[t]
\centering
\begin{tikzpicture}[
    req/.style={circle,draw,thick,minimum size=5mm,inner sep=0pt},
    set/.style={circle,draw,thick,minimum size=4.5mm,inner sep=0pt},
    Ain/.style={circle,draw,fill=black,minimum size=4.5mm,inner sep=0pt},
    edge/.style={thin}
]


\node[set] (s1) at (0,0) {};
\node at (0,-0.48) {$1$};

\node at (1.0,0) {$\cdots$};

\node[Ain] (a1) at (2.0,0) {};
\node[Ain] (a2) at (2.9,0) {};

\node at (3.9,0) {$\cdots$};

\node[set] (sd) at (5.2,0) {};
\node at (5.2,-0.48) {$d$};

\draw[thick] (5.95,-0.65) -- (5.95,0.65);

\node[set] (r1) at (6.8,0) {};
\node[set] (r2) at (7.6,0) {};
\node[set] (r3) at (8.4,0) {};
\node[set] (r4) at (9.2,0) {};
\node[set] (r5) at (10.0,0) {};
\node[set] (r6) at (10.8,0) {};
\node[set] (r7) at (11.6,0) {};
\node[set] (r8) at (12.4,0) {};
\node[set] (r9) at (13.2,0) {};

\node at (6.8,-0.48) {$d+1$};
\node at (13.2,-0.48) {$m$};

\draw (0,-0.95) -- (5.2,-0.95);
\node at (2.6,-1.25) {$[d]$};

\draw (6.8,-0.95) -- (13.2,-0.95);
\node at (10.0,-1.25) {$[m]\setminus[d]$};

\node at (2.45,-0.55) {$A$};


\node[req] (u11) at (1.7,4.0) {};
\node[req] (u12) at (2.9,4.0) {};
\node[req] (u13) at (4.1,4.0) {};
\node[req] (u14) at (5.3,4.0) {};

\node at (3.5,4.7) {};

\node[above=1pt] at (u11.north) {\scriptsize $(t,A,1)$};
\node[above=1pt] at (u12.north) {\scriptsize $(t,A,2)$};
\node[above=1pt] at (u13.north) {\scriptsize $(t,A,3)$};
\node[above=1pt] at (u14.north) {\scriptsize $(t,A,4)$};

\foreach \u in {u11,u12,u13,u14}{
    \draw[edge] (\u) -- (a1);
    \draw[edge] (\u) -- (a2);
}

\draw[edge] (u11) -- (r4);
\draw[edge] (u11) -- (r7);

\draw[edge] (u12) -- (r5);
\draw[edge] (u12) -- (r8);

\draw[edge] (u13) -- (r6);
\draw[edge] (u13) -- (r8);
\draw[edge] (u13) -- (r9);

\draw[edge] (u14) -- (r3);
\draw[edge] (u14) -- (r7);
\draw[edge] (u14) -- (r9);


\node[req] (u21) at (8.0,4.0) {};
\node[req] (u22) at (9.4,4.0) {};

\node at (8.7,4.7) {};

\node[above=1pt] at (u21.north) {\scriptsize $(t+1,A,1)$};
\node[above=1pt] at (u22.north) {\scriptsize $~~~~~(t+1,A,2)$};

\foreach \u in {u21,u22}{
    \draw[edge] (\u) -- (a1);
    \draw[edge] (\u) -- (a2);
}

\draw[edge] (u21) -- (r1);
\draw[edge] (u21) -- (r2);
\draw[edge] (u21) -- (r6);

\draw[edge] (u22) -- (r2);
\draw[edge] (u22) -- (r5);
\draw[edge] (u22) -- (r9);

\end{tikzpicture}

\caption{Requests from two consecutive levels for a fixed subset $A\subseteq[d]$. The set $A$ is the common incidence pattern on the first $d$ sets: every request shown is covered by exactly the sets indexed by $A$ inside $[d]$. The auxiliary incidences $R_{t,A,\ell}\subseteq[m]\setminus[d]$ distinguish different requests with the same level and the same set $A$. The number of requests decreases from level $t$ to level $t+1$.}
\label{fig:random-order-levels}
\end{figure}

	\randomorder*
	
	\begin{proof}
		We first define several parameters. Define $d:=\left\lfloor\left(\frac{m}{11}\right)^{2/3}\right\rfloor$, where $m$ is sufficiently large; the last $m-d$ sets are {\em dummy sets} whose role is to avoid having element copies. Let $k:=\lfloor\sqrt d\rfloor$, $q:=\lceil\sqrt d\rceil$, $r:=\lfloor d/32\rfloor$, and $T:=4q$. Also define $s:=\lceil100\cdot2^rd\log d\rceil$ and $D:=1+\lceil100s2^dT\rceil$. Finally, for $t<T$, let $L_t:=D^{T-t}$ be the number of requests of each type on {\em level} $t$, and let $z_t:=s/L_t$ be its priority threshold; define $N:=2^d\sum_{t<T}L_t$ as an upper bound on the total number of requests.

        Then, $d=\Theta(m^{2/3})$ and $m-d=\Theta(d^{3/2})$. Moreover, by the definitions of $N$ it holds that $\log_2 N \le d+\log_2 T+T\log_2 D$. Thus, by the definitions of the parameters $D,T$, and $q$: 
		\[
		\begin{aligned}
			\log_2 N
			\le d+o\left(d^{3/2}\right)+(4+o(1))\sqrt d
			\left(\frac{33}{32}+o(1)\right)d 
			=\left(\frac{33}{8}+o(1)\right)d^{3/2}
			<\frac{m-d}{2}.
		\end{aligned}
		\]
		It follows that $N<2^{m-d}$. 
		For every $t<T$ and every nonempty $A\subseteq[d]$, create $L_t$ requests of $(t,A)$, whose incidence patterns will be defined later in the proof. 
		For every such request $(t,A,\ell)$ 
		define a random subset $R_{t,A,\ell}\subseteq[m]\setminus[d]$, where all of these subsets
		are chosen jointly and uniformly among all choices that are pairwise distinct. 
		Indeed, a pairwise distinct selection exists since $N<2^{m-d}$. Moreover, give each request $(t,A,\ell)$ an independent uniform priority $p_{t,A,\ell} \sim \textnormal{Unif} [0,1]$, inducing a uniform random order of requests.
		
		Let $\mathcal E$ be the {\em bad} event that, for some $t<T$, a request corresponding to level $j>t$ has priority at most $z_t$. Furthermore,
		let $\mathcal B$ be the {\em good} event that (i) $\mathcal E$ does not hold and (ii) for every $t<T$, nonempty $A\subseteq[d]$, and $Y\subseteq[m]\setminus[d]$ with $|Y|\le r$, some request $(t,A,\ell)$ with $p_{t,A,\ell}\le z_t$ satisfies $R_{t,A,\ell}\cap Y=\emptyset$. 
		Intuitively, $\mathcal E$ is the event that a later level arrives too early, while $\mathcal B$ ensures that, no matter which (at most $r$) auxiliary sets from $[m] \setminus [d]$ the algorithm has bought, some level $t$ request cannot be covered by them and requires a set from $A\subseteq[d]$, which reduces to the worst case construction.
		We prove next that this good event $\mathcal{B}$ happens with high probability. 
		
		\begin{claim}
			$ \Pr[\mathcal B] > 9/10$. 
		\end{claim}
		\begin{claimproof}
			Fix $t<T$, nonempty $A\subseteq[d]$, and $Y\subseteq[m]\setminus[d]$ with $|Y|\le r$. For any $\ell\in[L_t]$ we have $\Pr[R_{t,A,\ell}\cap Y=\emptyset]=2^{-|Y|}$. Moreover, sampling the sets $R_{t,A,\ell}$ without replacement can only decrease the probability that all requests $(t,A,\ell)$ with $p_{t,A,\ell}\le z_t$ intersect $Y$, compared to independent sampling. Hence, since each request $(t,A,\ell)$ has probability $z_t2^{-|Y|}$ of having priority at most $z_t$ and satisfying $R_{t,A,\ell}\cap Y=\emptyset$ it holds that
			\begin{equation}
				\label{eq:z_t}
				\Pr\!\Big[
				R_{t,A,\ell}\cap Y\neq\emptyset
				~~\forall \ell \in [L_t]~\text{ s.t. }~p_{t,A,\ell}\le z_t
				\Big]
				\le (1-z_t2^{-|Y|})^{L_t}
				\le e^{-s2^{-r}}.
			\end{equation} The last inequality uses $z_tL_t=s$ and $|Y|\le r$.
			In addition, 
			since $\sum_{j>t}L_j<L_t/(D-1)$ it holds that
			\begin{equation}
				\label{eq:E}
			\Pr[\mathcal E]
			<\sum_{t<T}z_t2^d\sum_{j>t}L_j
			<\sum_{t<T}\frac{s}{L_t}2^d\frac{L_t}{D-1}
			=\frac{Ts2^d}{D-1}
			\le\frac1{100}.
			\end{equation}

			A union bound over all $t<T, A \subseteq [d]$ and $Y\subseteq[m]\setminus[d]$ with $|Y|\le r$ and \eqref{eq:z_t} give an upper bound on the probability of the complement event $\mathcal B^c$: 
			\[
			\begin{aligned}
				\Pr[\mathcal B^c]
				&\le \Pr[\mathcal E]
				+2^dT \sum^r_{i = 0} \binom{m-d}{i} e^{-s2^{-r}} \le \Pr[\mathcal E]+2^dT\left(\frac{e(m-d)}{r}\right)^r e^{-s2^{-r}}\\
				&\le \frac1{100}
				+\exp\!\left(
				d\log 2+\log T
				+r\log\frac{e(m-d)}{r}
				-s2^{-r}
				\right)
				=\frac1{100}+e^{-\Omega(d\log d)}
				<\frac1{10}
			\end{aligned}
			\]
			for a sufficiently large $m$.  
		The second inequality uses the standard bound
$
\sum_{i=0}^r \binom{m-d}{i}\le \left(\frac{e(m-d)}{r}\right)^r.
$
			The third inequality uses \eqref{eq:E}. The last equality uses $T=O(\sqrt d)$, $r=\Theta(d)$, $m-d=\Theta(d^{3/2})$, and $s2^{-r}\ge100d\log d$, so the exponent is $-\Omega(d\log d)$.
			Hence, $
			\Pr[\mathcal B]>\frac9{10}$. 
		\end{claimproof}

		We now define the incidence patterns of the requests using a random filtration system, similarly to the worst-case proof. 
		Fix, by \Cref{thm:union-expansion-main}, a $(d,k,q,r)$-union-expanding family $\mathcal C^{(0)}\subseteq\binom{[d]}k$ of size $2^T$, and for each $t < T$ choose $\mathcal C^{(t+1)}$ uniformly among the halves of $\mathcal C^{(t)}$. Recall the forcing batches $F_r(\mathcal C^{(t)})$ as in the worst-case proof. 
		Define  the $\ell$-th request of $(t,A)$ as the element $u_{\rho}$, where
		\[
		\rho=
		\begin{cases}
			A\cup R_{t,A,\ell}, & \text{if } u_A\in F_r(\mathcal C^{(t)}),\\
			[d]\cup R_{t,A,\ell}, & \text{otherwise}.
		\end{cases}
		\]
		Since the subsets $R_{t,A,\ell}$ are pairwise distinct and contained in $[m]\setminus[d]$, all requested elements have distinct incidence patterns; hence, there are no element copies. The requests in the above ``otherwise" case are dummy requests, included so that the event $\mathcal B$ is defined independently of the random filtration, simplifying the analysis. Assigning these dummy requests incidence $[d]\cup R_{t,A,\ell}$ ensures they impose no additional constraint on the optimal solution.

		To use Yao's minimax principle, fix a deterministic algorithm. Let $Z_t$ be the sets purchased by the algorithm after all requests with priority at most $z_t$ have arrived, and let $Z_T$ be its final set of purchases.
		Since $z_0<z_1<\cdots<z_{T-1}$, we have
		$
		Z_0\subseteq Z_1\subseteq\cdots\subseteq Z_{T-1}\subseteq Z_T.
		$
		Fix $t<T$ and nonempty $A\subseteq[d]$, and suppose that $|Z_t|\le r$ and $u_A\in F_r(\mathcal C^{(t)})$.
        We next prove that the good event implies that $Z_1,\ldots,Z_T$ satisfy the conditions of \Cref{lem:filtration}. 

        \begin{claim}
        \label{claim:conditions}
Conditioned on $\mathcal B$, the sequence $Z_0\subseteq\cdots\subseteq Z_T$ satisfies:
\begin{enumerate}
    \item For every $t<T$, $Z_t$ is determined before $\mathcal C^{(t+1)}$ is chosen.
    \item Conditioned on the history up to time $t$, $\mathcal C^{(t+1)}$ remains a uniformly random half of $\mathcal C^{(t)}$.
    \item Whenever $|Z_t|\le r$, there exists $W_t\in\mathcal C^{(t)}$ with $W_t\subseteq Z_t$.
\end{enumerate}
\end{claim}

\begin{proof}
For items 1 and 2, because $\mathcal B$ is independent of the
    filtration and $\mathcal B\subseteq\mathcal E^c$, conditioned on
    $\mathcal B$, for every $t<T$, $Z_t$ is determined before
    $\mathcal C^{(t+1)}$ is chosen, while $\mathcal C^{(t+1)}$ remains
    a uniformly random half of $\mathcal C^{(t)}$.
    For item 3, fix $t<T$ such that $|Z_t|\le r$. For every
    $u_A\in F_r(\mathcal C^{(t)})$, let
    $Y:=Z_t\cap([m]\setminus[d])$. Since $|Y|\le r$, by $\mathcal B$
    there exists a request $(t,A,\ell)$ with priority at most $z_t$
    such that $Y\cap R_{t,A,\ell}=\emptyset$. Since $Z_t$ covers this
    request, necessarily $Z_t\cap A\neq\emptyset$. Thus,
    $Z_t\cap[d]$ covers every element of $F_r(\mathcal C^{(t)})$.
    Hence, \Cref{lem:forcing-batch} gives
    $W_t\in\mathcal C^{(t)}$ with
    $W_t\subseteq Z_t\cap[d]\subseteq Z_t$.
\end{proof}

         By \Cref{claim:conditions}, conditioned on $\mathcal B$, the sequence $Z_0\subseteq\cdots\subseteq Z_T$ satisfies the hypothesis of \Cref{lem:filtration}. 
        Thus, \Cref{lem:filtration} gives, by viewing $\mathcal C^{(0)}\subseteq\binom{[d]}k$ as a $(m,k,q,r)$-union-expanding family in $\binom{[m]}k$, that 
\[
\mathbb E[|Z_T|\mid\mathcal B]
\ge
(1-e^{-q/4})r.
\]
Therefore,
		\[
		\mathbb E[|Z_T|]
		\ge
		\Pr[\mathcal B] \cdot \mathbb E[|Z_T|\mid\mathcal B]
		\ge
		\frac9{10}(1-e^{-q/4})r.
		\]
		Since the unique candidate cover in $\mathcal C^{(T)}:=\{W^\star\}$ belongs to $\mathcal C^{(t)}$ for every $t$, it intersects every $A$ for which $u_A\in F_r(\mathcal C^{(t)})$, as well as the dummy requests $[d]$. Hence, $W^\star$ covers every request and $\OPT\le k$. Since $r=\Theta(d)$ and $k=\Theta(\sqrt d)$, the resulting competitive ratio is $\Omega(\sqrt d)=\Omega(m^{1/3})$. 
		Yao's minimax principle therefore implies an $\Omega(m^{1/3})$ lower bound for randomized algorithms in the random order. 
	\end{proof}
	
	\paragraph{With element copies.}
	If element copies are allowed, the same argument gives the stronger $\Omega(\sqrt m)$ bound. We take $d=m$ and replace the elements of $(t,A)$ by distinct copies of each incidence pattern $A$. Using
	$
	n\le 2^m \cdot N=\exp(O(m^{3/2}))
	$
	elements, with high probability every relevant incidence pattern has a copy appearing before its level threshold, while no later level request appears too early. Thus, the same filtration argument gives expected cost $\Omega(m)$, whereas $\OPT=O(\sqrt m)$, yielding an $\Omega(\sqrt m)$ lower bound on the competitive ratio.
	
	\subsection{Bounded Space Lower Bound}
	\label{sec:space-proof}
	
	We now prove \Cref{thm:space-lb}. The proof is based on a reduction from one-way set-disjointness~\cite{ChakrabartiKhotSun03}.
    \paragraph{Multy-party one-way set-disjointness.} In the $q$-party one-way set-disjointness problem over universe $[N]$, player $1\le i\le q$ receives a set $S_i\subseteq[N]$, under the promise that one of the following cases hold: (i) the sets $S_1,\ldots,S_q$ are pairwise disjoint, or (ii) there is a unique element contained in all $q$ sets and the sets are pairwise disjoint if this element is removed. The players communicate in the following order: player $i$ sends a message to player $i+1$, and the last player $q$ outputs the answer. 
		The communication cost of a protocol is the maximum total number of bits sent between consecutive players.
		By the work of Chakrabarti, Khot, and Sun~\cite{ChakrabartiKhotSun03}, every randomized one-way protocol, even if all players have access to shared randomness, that distinguishes these two cases on any input with probability at least $2/3$ has communication cost $\Omega(N/q)$.
    
    \vspace{1em}
    
The proof of \Cref{thm:space-lb} reduces $q$-party one-way set disjointness to online set cover. Each element $s$ of the disjointness universe is represented by a candidate cover $W_s$ from a union-expansion family constructed using the second set of parameters in \Cref{thm:union-expansion-main}. Thus, each player replaces its input set $S_i$ by the corresponding candidate family ${W_s:s\in S_i}$ and presents its forcing batch (defined by Lemma~\ref{lem:forcing-batch}) to the online algorithm. Technically, we also add a distinct candidate $V_i$ to each player's family, ensuring that its forcing batch can always be covered. In the pairwise-disjoint case, union-expansion forces a large cost; conversely, if all sets share an element $s^\star$, then $W_{s^\star}$ covers every forcing batch, so $\OPT$ is small. Recall that memory is the information retained between requests; past requests cannot be reread unless explicitly stored in this memory.

	\spacelb*
	
	\begin{proof}
We start by constructing a union-expanding family. Let $d \ge 1$ be some fixed integer and assume that $m$, the number of sets, is sufficiently large. Let
		$
		k:=\left\lceil2(d+3)\log m\right\rceil$
		$q:=\left\lfloor\frac{m}{k}\right\rfloor$,
		and $r:=\left\lfloor\frac m{32}\right\rfloor$. Also
		define $N:=m^{d+3}$ and $M:=N+q$. 
		By \Cref{thm:union-expansion-main}, there is an $(m,k,q,r)$-union-expanding family $\mathcal{C}$ of size $M = N+q$.
		Label the members of $\mathcal{C}$ arbitrarily as
		$$
		\mathcal C=\{W_1,\ldots,W_N,V_1,\ldots,V_q\}.$$
    
		We reduce $q$-party one-way set-disjointness to online set cover as follows. Fix a randomized algorithm $\mathcal A$ for online set cover such that for all sufficiently large $m$ it holds that $\mathcal A$ uses at most $m^d$ bits of memory.
		Suppose, toward a contradiction, that $\mathcal A$ is $\alpha$-competitive for some $\alpha<r/(3k)$. 
        Now,
        consider an instance of $q$-party one-way set-disjointness with universe $[N]$ where each player $1\le i\le q$ receives a set $S_i \subseteq [N]$. The $q$ players use public randomness to simulate $\mathcal A$ as follows. 
        
        For each $1\le i\le q$, player $i$ defines
		$
		\mathcal D_i:=\{W_s:s\in S_i\}\cup\{V_i\}.
		$
		For $i=1$, player $1$ starts the simulation of $\mathcal A$ from the initial memory state of $\mathcal A$ with no sets purchased. For every $1<i\le q$, player $i$ receives from player $i-1$ two objects: the memory contents of $\mathcal A$ immediately after processing previous requests, and the vector indicating which of the $m$ sets have been purchased by $\mathcal A$ up to that point. Starting from this memory state, player $i$ simulates $\mathcal A$, request by request on the batch $F_r(\mathcal D_i)$, as defined by Lemma~\ref{lem:forcing-batch}. Whenever $\mathcal A$ purchases a set during this simulation, player $i$ updates the vector of purchased sets accordingly. If $i<q$, after all requests in $F_r(\mathcal D_i)$ have been processed, player $i$ sends to player $i+1$ the resulting memory contents of $\mathcal A$ together with the resulting vector of purchased sets.
		Let $X$ be the collection of all sets purchased by the algorithm $\mathcal{A}$ at termination. The last player outputs ``pairwise disjoint'' if $|X|>r$ and ``not disjoint'' otherwise.  
		
		Each of the first $q-1$ players sends at most $m^d$ bits describing the memory state of $\mathcal A$ and $m$ bits describing the vector of purchased sets. Hence, the total communication is at most
		$
		(q-1)(m^d+m)\le 2m^{d+1},
		$
		since $q\le m$ and $d\ge1$. We consider two cases for the protocol input. 
		
		\begin{enumerate}
			
			\item[(i)] The $(S_i)_i$ are pairwise disjoint. Then, $\mathcal D_1,\ldots,\mathcal D_q$ are pairwise disjoint. Assume, toward a contradiction, that $|X|\le r$. Since $X$ covers every batch $F_r(\mathcal D_i)$, by \Cref{lem:forcing-batch}, for every $i\in[q]$ there exists $U_i \in \mathcal D_i$ such that $U_i\subseteq X$. These $q$ candidates are distinct (but not necessarily disjoint); therefore, by the union-expansion property (Definition~\ref{def:union-expansion}) it holds that
			$
	|X|\ge\left|\bigcup_{i=1}^qU_i\right|>r.
			$ Contradiction.
			Thus, $|X|>r$, in the pairwise-disjoint case.
			
			\item[(ii)] The sets $(S_i)_i$ have a common element $s^\star$. Then, $W_{s^\star}\in\mathcal D_i$ for every $i$. As $|W_{s^\star}|=k\le r$, it follows from \Cref{lem:forcing-batch}  that $W_{s^\star}$ covers every forcing batch; therefore, $\OPT\le k$. Since $\mathbb E[|X|] \ge r \cdot \Pr[|X|>r]$,  
			\[
			\Pr[|X|>r]\le\frac{\mathbb E[|X|]}r\le \frac{\alpha\cdot \OPT}{r} \le \frac{\alpha k}{r}<\frac13.
			\]
		\end{enumerate}

		By the above, the simulation is a one-way protocol with success probability at least $2/3$. However,
		$
		\frac Nq\ge m^{d+2},
		$
		whereas its communication is at most $2m^{d+1}$, contradicting the $\Omega(N/q)$ lower bound of~\cite{ChakrabartiKhotSun03} for sufficiently large $m$. Consequently,
		$
		\alpha\ge\frac{r}{3k}=\Omega_d\!\left(\frac{m}{\log m}\right).
		$
	\end{proof}
	
	\section{Discussion and Open Questions}
	\label{sec:discussion}
	
	In this paper, we prove an $\Omega(\sqrt m)$ lower bound on the competitive ratio of every randomized information-theoretic algorithm for online set cover, matching the $\widetilde O(\sqrt m)$ upper bound of~\cite{Korman04} up to a logarithmic factor.
	The result holds in both the known and unknown instance models with $n=2^m-1$. 
	Our result also rules out an information-theoretic randomized algorithm that can match the competitive ratio of a fractional solution to the standard set cover LP. This shows that either a polylogarithmic dependence on $n$ or a polynomial dependence on $m$ is unavoidable.
	The central open question left for future work is to determine whether randomization in conjunction with unlimited computational power can break the $O(\log m\log n)$ upper bound of~\cite{AlonAwerbuchAzarBuchbinderNaor09} in the regime where $m$ and $n$ are polynomially related.
	More specifically, we leave the following question open.
	
	\begin{question}
		\label{q:log-mn}
		Does adversarial-order online set cover admit an information-theoretic randomized algorithm (against an oblivious adversary) with competitive ratio
		\(
		O(\log(mn))=O(\log m+\log n)?
		\)
	\end{question}
	
	\paragraph{Random-order.} in the random-order model, Gupta, Kehne, and Levin~\cite{GuptaKehneLevin21} showed that online set cover admits a polynomial-time randomized $O(\log(mn))$-competitive algorithm.
	\Cref{thm:random-order-lb} shows that, despite this improvement over adversarial order, a polynomial dependence of $\Omega(m^{1/3})$ remains unavoidable when $n$ is sufficiently large, and strengthens to $\Omega\left(\sqrt m\right)$ with element copies and  $n=\exp(O(m^{3/2}))$. 
	It remains open whether this can be improved to $\widetilde\Omega(\sqrt m)$ without element copies.
	
	\paragraph{Polynomial Space.} We also separate information-theoretic and polynomial-size memory algorithms, by showing a lower bound of $\Omega(m/\log m)$ for the latter, even when computation between requests is unlimited. An interesting question for future work is to answer whether there is an $\omega(\sqrt{m})$ lower bound for $\textnormal{poly}(m)$ time algorithms per request but allowing exponential space complexity. It is also unclear whether Theorem~\ref{thm:space-lb} holds in the random-order model. 
	
	\paragraph{Acknowledgments.}
	I. Doron-Arad is supported by grant NSF DMS-2031883 and Vannevar Bush Faculty Fellowship ONR-N00014-20-1-2826 (PI Mossel). J. Naor is Supported in part by ISF grant 3001/24 and United States– Israel BSF grant 2022418.

\end{document}